\documentclass{article}
\usepackage{amsmath,amssymb,bbm,theorem,ifthen}

\newcommand{\lp}[1]{\ifthenelse{\equal{#1}{0}}{(}{}\ifthenelse{\equal{#1}{1}}{\bigl(}{}\ifthenelse{\equal{#1}{2}}{\Bigl(}{}\ifthenelse{\equal{#1}{3}}{\biggl(}{}\ifthenelse{\equal{#1}{4}}{\Biggl(}{}\ifthenelse{\equal{#1}{5}}{\Biggl(}{}}
\newcommand{\rp}[1]{\ifthenelse{\equal{#1}{0}}{)}{}\ifthenelse{\equal{#1}{1}}{\bigr)}{}\ifthenelse{\equal{#1}{2}}{\Bigr)}{}\ifthenelse{\equal{#1}{3}}{\biggr)}{}\ifthenelse{\equal{#1}{4}}{\Biggr)}{}\ifthenelse{\equal{#1}{5}}{\Biggr)}{}}
\newcommand{\lbc}[1]{\ifthenelse{\equal{#1}{0}}{\{}{}\ifthenelse{\equal{#1}{1}}{\bigl\{}{}\ifthenelse{\equal{#1}{2}}{\Bigl\{}{}\ifthenelse{\equal{#1}{3}}{\biggl\{}{}\ifthenelse{\equal{#1}{4}}{\Biggl\{}{}\ifthenelse{\equal{#1}{5}}{\Biggl\{}{}}
\newcommand{\rbc}[1]{\ifthenelse{\equal{#1}{0}}{\}}{}\ifthenelse{\equal{#1}{1}}{\bigr\}}{}\ifthenelse{\equal{#1}{2}}{\Bigr\}}{}\ifthenelse{\equal{#1}{3}}{\biggr\}}{}\ifthenelse{\equal{#1}{4}}{\Biggr\}}{}\ifthenelse{\equal{#1}{5}}{\Biggr\}}{}}
\newcommand{\lba}[1]{\ifthenelse{\equal{#1}{0}}{\langle}{}\ifthenelse{\equal{#1}{1}}{\bigl\langle}{}\ifthenelse{\equal{#1}{2}}{\Bigl\langle}{}\ifthenelse{\equal{#1}{3}}{\biggl\langle}{}\ifthenelse{\equal{#1}{4}}{\Biggl\langle}{}\ifthenelse{\equal{#1}{5}}{\Biggl\langle}{}}
\newcommand{\rba}[1]{\ifthenelse{\equal{#1}{0}}{\rangle}{}\ifthenelse{\equal{#1}{1}}{\bigr\rangle}{}\ifthenelse{\equal{#1}{2}}{\Bigr\rangle}{}\ifthenelse{\equal{#1}{3}}{\biggr\rangle}{}\ifthenelse{\equal{#1}{4}}{\Biggr\rangle}{}\ifthenelse{\equal{#1}{5}}{\Biggr\rangle}{}}
\newcommand{\ve}[1]{\ifthenelse{\equal{#1}{0}}{|}{}\ifthenelse{\equal{#1}{1}}{\big|}{}\ifthenelse{\equal{#1}{2}}{\Big|}{}\ifthenelse{\equal{#1}{3}}{\bigg|}{}\ifthenelse{\equal{#1}{4}}{\Bigg|}{}\ifthenelse{\equal{#1}{5}}{\Bigg|}{}}
\newcommand{\lb}[1]{\ifthenelse{\equal{#1}{0}}{[}{}\ifthenelse{\equal{#1}{1}}{\bigl[}{}\ifthenelse{\equal{#1}{2}}{\Bigl[}{}\ifthenelse{\equal{#1}{3}}{\biggl[}{}\ifthenelse{\equal{#1}{4}}{\Biggl[}{}\ifthenelse{\equal{#1}{5}}{\Biggl[}{}}
\newcommand{\rb}[1]{\ifthenelse{\equal{#1}{0}}{]}{}\ifthenelse{\equal{#1}{1}}{\bigr]}{}\ifthenelse{\equal{#1}{2}}{\Bigr]}{}\ifthenelse{\equal{#1}{3}}{\biggr]}{}\ifthenelse{\equal{#1}{4}}{\Biggr]}{}\ifthenelse{\equal{#1}{5}}{\Biggr]}{}}
\newcommand{\srp}[3]{\ifthenelse{\equal{#1}{0}}{)^{#2}_{#3}}{}\ifthenelse{\equal{#1}{1}}{\bigr)^{#2}_{#3}}{}\ifthenelse{\equal{#1}{2}}{\Bigr)^{#2}_{#3}}{}
\ifthenelse{\equal{#1}{3}}{\biggr)^{#2}_{#3}}{}\ifthenelse{\equal{#1}{4}}{\Biggr)^{#2}_{#3}}{}\ifthenelse{\equal{#1}{5}}{\Biggr)^{#2}_{#3}}{}}
\newcommand{\srb}[3]{\ifthenelse{\equal{#1}{0}}{]^{#2}_{#3}}{}\ifthenelse{\equal{#1}{1}}{\bigr]^{#2}_{#3}}{}\ifthenelse{\equal{#1}{2}}{\Bigr]^{#2}_{#3}}{}
\ifthenelse{\equal{#1}{3}}{\biggr]^{#2}_{#3}}{}\ifthenelse{\equal{#1}{4}}{\Biggr]^{#2}_{#3}}{}\ifthenelse{\equal{#1}{5}}{\Biggr]^{#2}_{#3}}{}}
\newcommand{\srbc}[3]{\ifthenelse{\equal{#1}{0}}{{\}}^{#2}_{#3}}{}\ifthenelse{\equal{#1}{1}}{\bigr{\}}^{#2}_{#3}}{}\ifthenelse{\equal{#1}{2}}{\Bigr{\}}^{#2}_{#3}}{}
\ifthenelse{\equal{#1}{3}}{\biggr{\}}^{#2}_{#3}}{}\ifthenelse{\equal{#1}{4}}{\Biggr{\}}^{#2}_{#3}}{}\ifthenelse{\equal{#1}{5}}{\Biggr{\}}^{#2}_{#3}}{}}
\newcommand{\srba}[3]{\ifthenelse{\equal{#1}{0}}{\rangle^{#2}_{#3}}{}\ifthenelse{\equal{#1}{1}}{\bigr\rangle^{#2}_{#3}}{}\ifthenelse{\equal{#1}{2}}{\Bigr\rangle^{#2}_{#3}}{}
\ifthenelse{\equal{#1}{3}}{\biggr\rangle^{#2}_{#3}}{}\ifthenelse{\equal{#1}{4}}{\Biggr\rangle^{#2}_{#3}}{}\ifthenelse{\equal{#1}{5}}{\Biggr\rangle^{#2}_{#3}}{}}
\newcommand{\sve}[3]{\ifthenelse{\equal{#1}{0}}{|^{#2}_{#3}}{}\ifthenelse{\equal{#1}{1}}{\bigr|^{#2}_{#3}}{}\ifthenelse{\equal{#1}{2}}{\Bigr|^{#2}_{#3}}{}
\ifthenelse{\equal{#1}{3}}{\biggr|^{#2}_{#3}}{}\ifthenelse{\equal{#1}{4}}{\Biggr|^{#2}_{#3}}{}\ifthenelse{\equal{#1}{5}}{\Biggr|^{#2}_{#3}}{}}

\newcommand{\re}{\mathcal{R}e}

\newcommand{\disqueb}[1]{\ensuremath{D_{#1}}}
\newcommand{\demidisqueb}[1]{\ensuremath{D_{#1}^+}}

\newcommand{\disque}[1]{{\disqueb{#1}}}
\newcommand{\demidisque}[1]{{\demidisqueb{#1}}}

\newcommand{\AMSclass}[1]{{\textbf{A.M.S. subject classification:} #1}}
\newcommand{\assign}{:=}
\newcommand{\keywords}[1]{{\textbf{Keywords:} #1}}
\newcommand{\nonesep}{}
\newcommand{\tmdummy}{$\mbox{}$}
\newcommand{\tmmathbf}[1]{\ensuremath{\boldsymbol{#1}}}
\newcommand{\tmop}[1]{\ensuremath{\operatorname{#1}}}
\newcommand{\tmscript}[1]{\text{\scriptsize{$#1$}}}
\newcommand{\tmtextbf}[1]{{\bfseries{#1}}}
\newcommand{\tmtextit}[1]{{\itshape{#1}}}
\newcommand{\tmtextup}[1]{{\upshape{#1}}}
\newenvironment{itemizedot}{\begin{itemize} }{\end{itemize}}
\newenvironment{proof}{\noindent\textbf{Proof\ }}{\hspace*{\fill}$\Box$\medskip}
\numberwithin{equation}{section}  
\numberwithin{figure}{section}  
\newtheorem{theorem}{Theorem}[section]
\newtheorem{lemma}[theorem]{Lemma}
\newtheorem{proposition}[theorem]{Proposition}

\newtheorem{remark}[theorem]{Remark}

\begin{document}

\title{A deformation formula for the heat kernel\thanks{This paper has been written using the GNU TEXMACS scientific text editor.}
\thanks{\keywords{heat kernel,
quantum mechanics, Wiener integral, Feynman integral, Euler-Lagrange
equations, Hamiltonian systems, non-autonomous}; \AMSclass{35K08, 81Q30, 70H03,
70H05}}}\author{Thierry Harg\'e}\date{January 23, 2013}\maketitle

\begin{abstract}
  Let $P_0$ be a time-dependent partial differential operator acting on
  functions defined on $\mathbbm{C}^{\nu}$, quadratic with respect to
  $\partial_{x_1}, \ldots, \partial_{x_{\nu}}, x_1, \ldots, x_{\nu}$. Let $c$
  be a matrix-valued regular potential. Under suitable conditions, we give an
  ``explicit'' expression of ``the'' heat kernel associated to $P_0 + c$ for
  small $|t|$, $t \in \mathbbm{C}$, $\re t \geqslant 0$, $x, y \in
  \mathbbm{C}^{\nu}$.
\end{abstract}

\section{\label{ella}Introduction}

Let $\nu \geqslant 1$. For $j, k = 1, \ldots, \nu$, let $A_{j, k}$, $B_{j,
k}$, $C_{j, k}$ be complex functions analytic in a neighbourhood $U$ of the
origin. Let
\[ P_0 \assign \sum_{j, k = 1}^{\nu} A_{j, k} (t) \lp{2} \partial_{x_j} +
   \sum_{l = 1}^{\nu} B_{j, l} (t) x_l \rp{2} \lp{2} \partial_{x_k} + \sum_{l'
   = 1}^{\nu} B_{k, l'} (t) x_{l'} \rp{2} - \]
\begin{equation}
  \label{ella2} \sum_{j, k = 1}^{\nu} C_{j, k} (t) x_j x_k .
\end{equation}
As well-known, under suitable assumptions on the matrix $A$, the equation
\begin{equation}
  \label{ella4} \left\{ \begin{array}{l}
    \partial_t u = P_0 u\\
    \\
    u_{} |_{t = 0^+} =_{} \delta_{x = y}
  \end{array} \right.
\end{equation}
admits an explicit solution
\begin{equation}
  \label{ella5} p_t^0 (x, y) \assign \frac{1}{(4 \pi \Delta t)^{\nu / 2}} e^{-
  \frac{1}{t} \Phi_0 (x, y, t)} .
\end{equation}
Here $\Phi_0$ denotes a polynomial of total degree 2 with respect to $x, y$,
whose coefficients are analytic near 0; $\Delta$ is the determinant of the
matrix $\lp{1} A_{j, k} (0) \srp{1}{}{1 \leqslant j, k \leqslant \nu}$. Let
$c$ be a regular square matrix-valued function defined on $U \times
\mathbbm{C}^{\nu}$. Let $p_t (x, y)$ be a solution of
\begin{equation}
  \label{ella6} \left\{ \begin{array}{l}
    \partial_t u = \lp{1} P_0 + c (t, x) \rp{1} u\\
    \\
    u_{} |_{t = 0^+} =_{} \delta_{x = y}
  \end{array} \right.
\end{equation}
and let $p_t^{\tmop{conj}} (x, y)$ be defined by
\[ p_t (x, y) = p_t^0 (x, y) p_t^{\tmop{conj}} (x, y) . \]
Roughly speaking, if the autonomous case is considered for the sake of
simplicity, $p_t^{\tmop{conj}}$ may be related to two transition amplitudes:
\[ p_t^{\tmop{conj}} (x, y) = \frac{\left\langle y| \exp \lp{1} t (P_0 + c)
   \rp{1} |x \right\rangle}{\left\langle y| \exp (t P_0) |x \right\rangle} .
\]
First let us assume that $P_0 = \partial_{x_1}^2 + \cdots +
\partial_{x_{\nu}}^2$ (free case) or, more generally, $P_0 = \partial_{x_1}^2
+ \cdots + \partial_{x_{\nu}}^2 - \lambda (x_1^2 + \cdots + x_{\nu}^2)$,
$\lambda \in \mathbbm{R}$ (harmonic case). Under strong assumptions on the
scalar potential $c$, $p^{\tmop{conj}}$ is Borel summable with respect to $t$
[Ha4]. The same work is done in [Ha5] in the free case, but with a vector
potential instead of a scalar one. A so-called deformation formula
(respectively its vector potential version) which gives a convenient
representation of $p_t^{\tmop{conj}}$ is used.

The main aim of this paper is to explore the limits of this formula by
considering the non-autonomous case. This explains the choice of the operator
$P_0$ in (\ref{ella2}). Our deformation formula is given in Theorem
\ref{lola20}. We do not attempt to give a uniqueness statement for the
definition of the heat kernel and we refer to [Ha7] for precise statements
about this question.

This formula is related to Wiener and Feynman integrals [It, A-H]. As in these
references, we write the potential as the Fourier transform of a Borel
measure. See [Ha4] for more details about the relationship between this
formula and Wiener or Feynman integrals.

The shape of the formula can be explained by using a heuristic Wiener
representation of $p_t (x, y)$ and Wick's theorem (see [Ha4, Appendix]).
However we prove the deformation formula by working directly on the equations
satisfied by $p_t^{\tmop{conj}} (x, y)$ without attempting to obtain an
expression of $p_t (x, y)$. This formula uses a so-called deformation matrix.
This matrix, in the autonomous case, is considered in [Ge-Ya, On].

By the heuristic method, it is easy to see that the construction of this
matrix involves a ``propagator'' (defined as in quantum field theory). Here we
first give another definition of this object and we prove a posteriori that it
verifies the propagator equation (see section \ref{emma15.5}).

The shape of the operator $P_0$ implies that $p_t^0 (x, y)$ can be written
explicitly using the solution of a classical Hamiltonian system. The
deformation matrix also depends on this solution, hence, by the deformation
formula, the expression of $p^{\tmop{conj}}_t (x, y)$ only involves objects
related to this Hamiltonian system.

We assume in Theorem \ref{lola20} that the functions $A$, $B$ and $C$ satisfy
a reality assumption (see (\ref{lola22})), implying that $P_0 |_{t \in
i\mathbbm{R}}$ is symmetric with respect to the $L^2$-inner product. This
assumption is natural if the Schr\"odinger kernel is considered.

One can certainly establish a deformation formula in the case of a vector
potential perturbation of $P_0$ (instead of a scalar potential one) and give a
Borel summability statement for the small time expansion of $p_t^{\tmop{conj}}
(x, y)$.

\section{\label{lola}Notation and main results}

For $z = |z|e^{i \theta} \in \mathbbm{C}$, $\theta \in [- \pi / 2, \pi / 2]$,
let $z^{1 / 2} \assign |z|^{1 / 2} e^{i \theta / 2}$. For $T > 0$, let
$\disque{T} \assign \{z \in \mathbbm{C} | |z| < T\}$, $\demidisque{T} \assign
\{z \in \disque{T} | \re (z) > 0\}$ and $\bar{D}^+_T \assign \{z \in
\disque{T} | \re (z) \geqslant 0\}$. For $\lambda, \mu \in \mathbbm{C}^{\nu}$,
we denote $\lambda \cdot \mu \assign \lambda_1 \mu_1 + \cdots + \lambda_{\nu}
\mu_{\nu}$, $\lambda^2 \assign \lambda \cdot \lambda$, $\bar{\lambda} \assign
( \bar{\lambda}_1, \ldots, \bar{\lambda}_{\nu})$, $| \lambda | \assign
(\lambda \cdot \bar{\lambda})^{1 / 2}$ and we extend the two first notations
to operators such as $\partial_x = (\partial_{x_1}, \ldots,
\partial_{x_{\nu}})$. Let $A = (A_{j, k})_{1 \leqslant j, k \leqslant \nu}$
with $A_{j, k} \in \mathbbm{C}$. If $x, y \in \mathbbm{C}^{\nu}$, $\sum_{j, k}
A_{j, k} x_j y_k$ is denoted by $x \cdot A y$ or $A \cdot x \otimes y$ if $A$
is symmetric. We set $\text{}^{\mathfrak{t}} A$ for the transpose of the
matrix $A$ and $|A|_{\infty} \assign \sup_{|x| = 1} |A x|$. In what follows,
we shall consider a potential function defined on $\disque{T} \times
\mathbbm{C}^{\nu}$ with values in a finite dimensional space of square
matrices, say $\mathcal{M}$. We always use on $\mathcal{M}$ a norm $| \cdot |$
such that $|AB| \leqslant |A\|B|$ for $A, B \in \mathcal{M}$ and $|
\mathbbm{1} | = 1$ ($\mathbbm{1}$ denotes the unitary matrix). Let $\Omega$ be
an open domain in $\mathbbm{C}^m$ and let $F$ be a complex finite dimensional
space. We denote by $\mathcal{A} (\Omega)$ the space of $F$-valued analytic
functions on $\Omega$, if there is no ambiguity on $F$. Let $T > 0$.
$\mathcal{C}^{\infty} \lp{1} \bar{D}^+_T, \text{$\mathcal{A}(\mathbbm{C}^{2
\nu})$} \rp{1}$ denotes the space of smooth functions defined on $\bar{D}^+_T$
with values in $\mathcal{A}(\mathbbm{C}^{2 \nu}$). We denote by
$\mathcal{C}^{\infty}_{b, 1} (i] - T, T [\times \mathbbm{R}^m)$ the space of
smooth $\mathcal{M}$-valued functions defined on $i] - T, T [\times
\mathbbm{R}^m$ such that
\[ f \in \mathcal{C}^{\infty}_{b, 1} (i] - T, T [\times \mathbbm{R}^m)
   \Leftrightarrow \]
\[ \forall (\alpha, \beta) \in \mathbbm{N} \times \mathbbm{N}^m, \exists C >
   0, \forall (t, x) \in i] - T, T [\times \mathbbm{R}^m, |
   \partial_t^{\alpha} \partial^{\beta}_x f (t, x) | \leqslant C (1 +
   |x|)^{\alpha} . \]
We always consider these spaces with their standard Frechet structure (the
semi-norms are eventually indexed by compact sets or differentiation order).

Let $\mathfrak{B}$ be the collection of all Borel sets on $\mathbbm{R}^m$. An
$F$-valued measure $\mu$ on $\mathbbm{R}^m$ is an $F$-valued function on
$\mathfrak{B}$ satisfying the classical countable additivity property \ [Ru].
Let $| \cdot |$ be a norm on $F$. We denote by $| \mu |$ the positive measure
defined \ by
\[ | \mu | (E) = \sup \sum_{j = 1}^{\infty} | \mu (E_j) | (E \in
   \mathfrak{B}), \]
the supremum being taken over all partition $\{E_j \}$ of $E$. In particular
$| \mu | (\mathbbm{R}^m) < \infty$. Note that $d \mu = h d| \mu |$ where $h$
is some $F$-valued function satisfying $|h| = 1$ $| \mu |$-a.e. If $f$ is an
$F$-valued measurable function on $\mathbbm{R}^m$ and $\lambda$ is a positive
measure on $\mathbbm{R}^m$ such that $\int_{\mathbbm{R}^m} |f| d \lambda <
\infty$, one may define an $F$-valued measure $\mu$, by setting $d \mu = f d
\lambda$. Then $d| \mu | = |f |d \lambda$.

Let $A$, $B$ and $C$ be some $\nu \times \nu$ complex matrix-valued analytic
functions defined on a neighbourhood of 0 in $\mathbbm{C}$. Let us assume that
the matrices $A$, $C$ are symmetric and that the matrix $A (0)$ is real
positive definite. The operator $P_0$ defined in (\ref{ella2}) can be
rewritten as
\begin{equation}
  \label{lola4} P_0 = A (t) \cdot (\partial_x + B (t) x)^2 - C (t) \cdot x
  \otimes x
\end{equation}
where
\[ A (t) \cdot (\partial_x + B (t) x)^2 \assign \sum_{j, k = 1}^{\nu} A_{j, k}
   (t) \lp{2} \partial_{x_j} + \sum_{l = 1}^{\nu} B_{j, l} (t) x_l \rp{2}
   \lp{2} \partial_{x_k} + \sum_{l' = 1}^{\nu} B_{j, l'} (t) x_{l'} \rp{2} .
\]
The fact that (\ref{ella4}) admits a solution as mentioned in the introduction
will be recovered later. Our main result gives a formula for the solution of a
perturbation of (\ref{ella4}). We need some classical objects associated to
the operator defined in (\ref{lola4}). For $t \in \mathbbm{C}$, $|t|$ small,
let $L$ be the following Lagrangian acting on $\mathbbm{C}^{\nu}$-valued
functions
\begin{equation}
  \label{lola10} L \assign \frac{1}{4} \dot{q} \cdot A^{- 1} \dot{q} + \dot{q}
  \cdot B q + q \cdot C q.
\end{equation}
The Euler-Lagrange equations associated to (\ref{lola10}) can be written
\begin{equation}
  \label{lola12} \ddot{q} = E \dot{q} + F q
\end{equation}
where
\begin{equation}
  \label{lola13} E \assign \dot{A} A^{- 1} + 2 A ( \text{}^{\mathfrak{t}} B -
  B) \text{ , \ } F \assign 4 A C - 2 A \dot{B} .
\end{equation}
Let $x, y \in \mathbbm{C}^{\nu}$. Let $t$ be a small positive number. We
denote by $q^{\natural}_t$ the solution of (\ref{lola12}) with the conditions
$q^{\natural}_t (0) = y$, $q^{\natural}_t (t) = x$ if it is uniquely defined.
Notice that $q^{\natural}_t$ can be expressed by
\begin{equation}
  \label{lola13.5} q^{\natural}_t = q^{\flat}_t x + q^{\sharp}_t y
\end{equation}
where the matrices-valued functions $q^{\flat}_t$, $q^{\sharp}_t$ are
respectively solutions of
\begin{equation}
  \label{lola14} \ddot{q} = E \dot{q} + F q \text{ \ \ \ (matrix-valued
  equation)}
\end{equation}
with the conditions
\[ q^{\flat}_t (0) = 0 \text{ , \ } q^{\flat}_t (t) = \mathbbm{1}, \]
respectively
\[ q^{\sharp}_t (0) = \mathbbm{1} \text{ , \ } \text{$\mathbbm{}$} \text{}
   \text{} \text{} \text{$q^{\sharp}_t (t) = 0$} . \]
Let $\tilde{q}^{\diamond}_t$ (with $\diamond = \natural, \flat, \sharp$) be
defined on $[0, 1]$ by the relation
\begin{equation}
  \label{lola14.5} q_t^{\diamond} (s) = \tilde{q}^{\diamond}_t \lp{1}
  \frac{s}{t} \rp{1} .
\end{equation}
Notice that $\tilde{q}^{\flat}_t$ is the solution of
\begin{equation}
  \label{lola15} \ddot{q} = t E (t s) \dot{q} + t^2 F (t s) q \text{ , \ } q
  (0) = 0 \text{ , \ } q (1) = \mathbbm{1}
\end{equation}
and that $\tilde{q}^{\sharp}_t$ is the solution of
\begin{equation}
  \label{lola15bis} \ddot{q} = t E (t s) \dot{q} + t^2 F (t s) q \text{ , \ }
  q (0) = \mathbbm{1} \text{ , \ } q (1) = 0.
\end{equation}
Actually (\ref{lola15}) (respectively (\ref{lola15bis})) admits a unique
solution for complex $t$ with small modulus, which provides the good
definition of $\tilde{q}^{\flat}_t (s)$ (respectively $\tilde{q}^{\sharp}_t
(s)$). Then, by (\ref{lola14.5}), one gets the existence and uniqueness of
$q^{\flat}_t (s), q^{\sharp}_t (s)$ for $s \in [0, t]$ and small positive $t$.
The following expressions play a central role in the statement of our main
result. For $s, s' \in [0, t]$, let $K_t (s, s')$ be the $\nu \times \nu$
matrix defined by
\begin{equation}
  \label{lola16} K_t (s, s') \assign \int_{s \vee s'}^t q^{\flat}_{\tau} (s) A
  (\tau) \text{}^{\mathfrak{t}} q^{\flat}_{\tau} (s') d \tau .
\end{equation}
For complex $t$ with small modulus and $s, s' \in [0, 1]$, we also denote
\begin{equation}
  \label{lola18} \tilde{K}_t (s, s') \assign \int_{s \vee s'}^1
  \tilde{q}^{\flat}_{t \tau} \lp{1} \frac{s}{\tau} \rp{1} A (t \tau)
  \text{}^{\mathfrak{t}} \tilde{q}^{\flat}_{t \tau} \lp{1} \frac{s'}{\tau}
  \rp{1} d \tau .
\end{equation}
If $t$ is real, positive and $s, s' \in [0, t]$, one gets
\[ K_t (s, s') = t \tilde{K}_t \lp{1} \frac{s}{t}, \frac{s'}{t} \rp{1} . \]
For $s = (s_1, \ldots, s_n) \in [0, 1]^n$, we define
\[ \tmmathbf{\tilde{K}}_t (s) \cdot \partial_z \otimes_n \partial_z \assign
   \sum_{j, k = 1}^n \partial_{z_j} \cdot \tilde{K}_t (s_j, s_k)
   \partial_{z_k} . \]
This differential operator acts on $\mathcal{A}(\mathbbm{C}^{\nu n})$.

We denote by $\bar{T}$ a positive number such that
(\ref{lola15})-(\ref{lola15bis}) admits a unique solution for $t \in
D_{\bar{T}}$ and the map $(t, s) \longmapsto q_t^{\diamond} (s)$ (with
$\diamond = \flat, \sharp$) is analytic near $D_{\bar{T}} \times [0, 1]$.

\begin{theorem}
  \label{lola20}Let $T_b > 0$. Let $f$ be a measurable function on
  $\disque{T_b} \times \mathbbm{R}^{\nu}$ with values in a complex finite
  dimensional space of square matrices, analytic with respect to the first
  variable. Let $\mu_{\ast}$ be a positive measure on $\mathbbm{R}^{\nu}$.
  Assume that for every $R > 0$
  \begin{equation}
    \label{lola21} \int_{\mathbbm{R}^{\nu}} e^{R| \xi |} \sup_{|t| < T_b} |f
    (t, \xi) |d \mu_{\ast} (\xi) < \infty .
  \end{equation}
  Let $\mu_t$ be the measure defined by $d \mu_t (\xi) = f (t, \xi) d
  \mu_{\ast} (\xi)$ and let
  \[ c (t, x) = \int_{\mathbbm{R}^{\nu}} \exp (ix \cdot \xi) d \mu_t (\xi) .
  \]
  Let $P_0$ be an operator as in (\ref{lola4}). Let us assume that each $g =
  A, i B, C$ satisfies:
  \begin{equation}
    \label{lola22} \text{The function } g|_{i\mathbbm{R}} \text{ is
    real-valued near } 0.
  \end{equation}
  Let $p^{\tmop{conj}}$ be defined by
  \begin{equation}
    \label{lola23} p^{\tmop{conj}} = \mathbbm{1} + \sum_{n \geqslant 1} v_n,
  \end{equation}
  where
  \[ v_n (t, x, y) \assign t^n \times \]
  \begin{equation}
    \label{lola24} \int_{0 < s_1 < \cdots < s_n < 1} \lb{2} e^{t
    \tmmathbf{\tilde{K}}_t (s) \cdot \partial_z \otimes_n \partial_z } c (s_n
    t, z_n) \cdots c (s_1 t, z_1) \rb{2} \ve{2}_{\tmscript{\begin{array}{l}
      z_1 = \tilde{q}^{\natural}_t (s_1)\\
      \ldots\\
      z_n = \tilde{q}^{\natural}_t (s_n)
    \end{array}}} d^n s.
  \end{equation}
  Then there exists $T_c > 0$ such that
  \[ p^{\tmop{conj}} \in \mathcal{A} ( \demidisque{T_c} \times \mathbbm{C}^{2
     \nu}) \cap \mathcal{C}^{\infty} \lp{1} \bar{D}^+_{T_c},
     \mathcal{A}(\mathbbm{C}^{2 \nu}) \rp{1} \cap \mathcal{C}^{\infty}_{b, 1}
     (i] - T_c, T_c [\times \mathbbm{R}^{2 \nu}) . \]
  The function $u : = p^0 \times p^{\tmop{conj}}$ is a solution of
  (\ref{ella6}).
\end{theorem}

We shall now give another useful expression of $p^{\tmop{conj}}$. Let $n
\geqslant 1$, $s = (s_1, \ldots, s_n) \in [0, 1]^n$ \ and $\xi = (\xi_1,
\ldots, \xi_n) \in \mathbbm{R}^{\nu n}$. Let
\[ \tilde{q}^{\natural}_t (s_{}) \cdot \xi : = \tilde{q}^{\natural}_t (s_1)
   \cdot \xi_1 + \cdots + \tilde{q}^{\natural}_t (s_n) \cdot \xi_n, \]
\begin{equation}
  \label{lola30} \tmmathbf{\tilde{K}}_t (s) \cdot \xi \otimes_n \xi \assign
  \sum_{j, k = 1}^n \xi_j \cdot \tilde{K}_t (s_j, s_k) \xi_k,
\end{equation}
\[ d^{\nu n} \mu^{\otimes}_{s t} (\xi) \assign d^{\nu} \mu_{s_n t} (\xi_n)
   \cdots d^{\nu} \mu_{s_1 t} (\xi_1) . \]
Then, we get

\begin{remark}
  {\tmdummy}
  
  \begin{equation}
    \label{lola32} v_n (t, x, y) = t^n \int_{0 < s_1 < \cdots < s_n < 1} \int
    e^{i \tilde{q}^{\natural}_t (s_{}) \cdot \xi} e^{- t
    \tmmathbf{\tilde{K}}_t (s) \cdot \xi \otimes_n \xi} d^{\nu n}
    \mu^{\otimes}_{s t} (\xi) d^n s.
  \end{equation}
\end{remark}

\begin{remark}
  Let us consider some examples.
  \begin{itemizedot}
    \item In the free case $A = \mathbbm{1}$, $B = C = 0$, we have (see also
    [Ha4])
    \[ \tilde{q}^{\natural}_t (s_{}) = y + s (x - y), \]
    \[ \tmmathbf{\tilde{K}}_t (s) \cdot \partial_z \otimes_n \partial_z =
       \sum_{j, k = 1}^n s_{j \wedge k} (1 - s_{j \vee k}) \partial_{z_j}
       \cdot \partial_{z_k} . \]
    Then
    \[ v_n (t, x, y) = t^n \int_{0 < s_1 < \cdots < s_n < 1} \int e^{i \lp{1}
       y + s (x - y) \rp{1} \cdot \xi} \times \]
    \[ \exp \lp{1} - t s (1 - s) \cdot_n \xi \otimes \xi \rp{1} d^{\nu n}
       \mu^{\otimes}_{s t} (\xi) d^n s \]
    where
    \[ s (1 - s) \cdot_n \xi \otimes \xi \assign \sum_{j, k = 1}^n s_{j
       \wedge k} (1 - s_{j \vee k}) \xi_j \cdot \xi_k . \]
    \item Let us assume that $A = \mathbbm{1}$, $C = 0$ and $B = - \frac{i}{2}
    \beta$ where the matrix $\beta$ is skew-symmetric, real and constant with
    respect to $t$. Then
    \[ \tilde{q}^{\natural}_t (s_{}) = e^{- i \beta t (1 - s)} \frac{\sin
       (\beta t s)}{\sin (\beta t)} x + e^{i \beta t s} \frac{\sin \lp{1}
       \beta t (1 - s) \rp{1}}{\sin (\beta t)} y, \]
    \[ \tilde{K}_t (s, s') = e^{i \beta t (s' - s)} \frac{\sin (\beta t s
       \wedge s') \sin \lp{1} \beta t (1 - s \vee s') \rp{1}}{\beta t \sin
       (\beta t)} . \]
  \end{itemizedot}
\end{remark}

\begin{remark}
  Let us make some comments on the functional spaces introduced in Theorem
  \ref{lola20}. The space $\mathcal{A} ( \demidisque{T_c} \times
  \mathbbm{C}^{2 \nu})$ allows one to consider the function $p^{\tmop{conj}}$
  as a solution of the (complex) heat equation. The space
  $\mathcal{C}^{\infty} \lp{1} \bar{D}^+_{T_c}, \mathcal{A}(\mathbbm{C}^{2
  \nu}) \rp{1}$ allows one to consider the function $p^{\tmop{conj}}$ as a
  solution of the Schr\"odinger equation, viewed as a limit case of the heat
  equation. Both spaces are local with respect to the space variables $x$ and
  $y$. The space $\mathcal{C}^{\infty}_{b, 1} (i] - T_c, T_c [\times
  \mathbbm{R}^{2 \nu})$, which gives information about global properties of
  the function $p^{\tmop{conj}}$ with respect to the space variables, provides
  a unicity statement (see [Ha7]).
\end{remark}

\section{\label{emma}Proofs of the results}

\subsection{Some properties of classical objects associated to $P_0$}

Let us recall why equation (\ref{ella4}) admits a solution as in
(\ref{ella5}). Let
\begin{equation}
  \label{emma2} S (x, y, t) \assign \int_0^t L|_{q = q^{\natural}_t} d s =
  \frac{1}{t} \int_0^1 \tilde{L} |_{q = \tilde{q}^{\natural}_t} d s =
  \frac{1}{t} \Phi (x, y, t)
\end{equation}
where
\[ \tilde{L} \assign \frac{1}{4} A^{- 1} (t s) \cdot \dot{q} \otimes \dot{q} +
   t \dot{q} \cdot B (t s) q + t^2 C (t s) \cdot q \otimes q \]
and
\begin{equation}
  \label{emma3} \Phi (x, y, t) \assign \int_0^1 \tilde{L} |_{q =
  \tilde{q}^{\natural}_t} d s.
\end{equation}
$\Phi$, as $\tilde{q}^{\natural}_t$, is analytic for complex $t$ with small
modulus. Since $\tilde{q}^{\natural}_t$ is linear with respect to $x, y$,
$\Phi$ is a polynomial of total degree 2 in $x, y$ and its coefficients are
analytic near 0. By classical theory, since $q^{\natural}_t$ verifies the
Euler-Lagrange equations, $S$ satisfies the eikonal equation
\[ \partial_t S + H|_{q = x, p = \partial_x S} = 0, \]
where
\[ H = \dot{q} \cdot \frac{\partial L}{\partial \dot{q}} - L = A \cdot (p - B
   q) \otimes (p - B q) - C \cdot q \otimes q. \]
Let us remark that
\begin{equation}
  \label{emma4} \partial_x S = p^{\natural}_t (t),
\end{equation}
where $p_t^{\natural} \assign \frac{\partial L}{\partial \dot{q}} |_{q =
q_t^{\natural}}$.

Putting $u = \lambda_t e^{- S}$ in (\ref{ella4}) shows us that the partial
differential equation in (\ref{ella4}) is equivalent to
\begin{equation}
  \label{emma6} \text{$\partial_t \lambda_t = \lp{1} - \frac{1}{t} A (t) \cdot
  \partial_x^2 \Phi + \gamma (t) \rp{1} \times \lambda_t$},
\end{equation}
where $\gamma (t) \assign \tmop{Tr} \lp{1} A (t) B (t) \rp{1}$. Differential
equations and boundary conditions satisfied by $\tilde{q}^{\flat}_t$ and
$\tilde{q}^{\sharp}_t$ involve that $\tilde{q}^{\natural}_t = y + s (x - y) +
t \chi (s, t, x, y)$ where $\chi$ is linear in $x, y$ with analytic
coefficients in $s, t$. Then
\begin{equation}
  \label{emma8} \Phi (x, y, t) = \frac{1}{4} A^{- 1} (0) \cdot (x - y)^2 + t
  \Psi (t, x, y),
\end{equation}
where $\Psi$ is a polynomial of total degree bounded by $2$, with respect to
$x, y$, with analytic coefficients in $t$. Since $\partial_x^2 \Phi$ only
depends on $t$,
\[ - \frac{1}{t} A (t) \cdot \partial_x^2 \Phi + \gamma (t) = - \frac{\nu}{2
   t} + \theta (t) \]
where the function $\theta$ is analytic near $0$. Then for every $\Delta > 0$
and every polynomial $K \in \mathbbm{C}[y]$,
\[ \lambda_t = \frac{1}{(4 \pi \Delta t)^{\nu / 2}} \exp \lp{2} \int_0^t
   \theta (s) d s + K (y) \rp{2} \]
satisfies (\ref{emma6}). Then, by (\ref{emma8}), the function
\[ u \assign \frac{1}{(4 \pi \Delta t)^{\nu / 2}} e^{\int_0^t \theta (s) d s -
   \Psi (t, x, y) + K (y)} e^{- \frac{1}{4 t} A^{- 1} (0) \cdot (x - y)^2} \]
is solution of the partial differential equation in (\ref{ella4}). Let us
choose
\[ \Delta \assign \det \lp{1} A_{j, k} (0) \srp{1}{}{1 \leqslant j, k
   \leqslant \nu} . \]
Then
\[ u|_{t = 0^+} = e^{- \Psi (0, y, y) + K (y)} \delta_{x = y} . \]
Let us choose $K (y) = \Psi (0, y, y) = \frac{1}{t} \Phi (y, y, t)$. Then $u$
is the solution of (\ref{ella4}) and, denoting this solution by $p^0$,
(\ref{ella5}) is satisfied where
\begin{equation}
  \label{emma8.5} \Phi_0 (x, y, t) \assign \Phi (x, y, t) - \Phi (y, y, t) +
  t^2 \int_0^1 \theta (t \nonesep s) d s.
\end{equation}
The following results will be useful.

\begin{lemma}
  \label{emma9.1}There exists $T_a \in] 0, \bar{T} [$ such that for every $s
  \in [0, 1]$, $(x, y) \in \mathbbm{C}^{2 \nu}$ and $t \in \disque{T_a}$,
  \begin{equation}
    \label{emma10} | \tilde{q}^{\sharp}_t (s) | \leqslant 2 (|x| + |y|) .
  \end{equation}
\end{lemma}

\begin{proof}
  By (\ref{lola15}), $| \text{$\tilde{q}^{\flat}_t$} (s) | \leqslant 2$ for
  complex $t$ with small modulus. Similarly, the same estimate holds for
  $\tilde{q}^{\sharp}_t$. Then (\ref{lola13.5}) implies (\ref{emma10}).
\end{proof}

\begin{proposition}
  For small positive number $t$, for every $(x, y) \in \mathbbm{C}^{2 \nu}$,
  $s \in [0, t]$ and $\alpha \in \{1, \ldots, \nu\}$
  \begin{equation}
    \label{emma14} \frac{1}{p^0} \lp{1} \partial_x + B (t) x \rp{1} p^0 = -
    \frac{1}{2} A^{- 1} (t) \dot{q}^{\natural}_t (t) .
  \end{equation}
  \begin{equation}
    \label{emma15} \lp{1} \partial_t + \dot{q}^{\natural}_t (t) \cdot
    \partial_x \rp{1} q^{\natural}_{t, \alpha} (s) = 0,
  \end{equation}
  $q_{t, \alpha}^{\natural}$ denoting the $\alpha$-coordinate of the vector
  $q_t^{\natural}$.
\end{proposition}

\begin{proof}
  Recall that $p^0 \assign \frac{1}{(4 \pi \Delta t)^{\nu / 2}} e^{-
  \frac{1}{t} \Phi_0}$. By (\ref{emma8.5}), $p^0 = \frac{1}{(4 \pi \Delta
  t)^{\nu / 2}} e^{- S + \Gamma (t, y)}$ where $\Gamma$ is a polynomial in $y$
  with analytic coefficients in $t$ near 0. Then, by (\ref{emma4}),
  \[ \frac{1}{p^0} (\partial_x + B (t) x) p^0 = - \partial_x S + B (t) x = -
     p_t^{\natural} (t) + B (t) q_t^{\natural} (t) . \]
  But $p = \frac{\partial L}{\partial \dot{q}} = \frac{1}{2} A^{- 1} \dot{q} +
  B q$. This proves (\ref{emma14}).
  
  Let $w (s) \assign \lp{1} \partial_t + \dot{q}^{\natural}_t (t) \cdot
  \partial_x \rp{1} q^{\natural}_t (s)$. Hence $w (0) = 0$ and $w (t) = \lb{1}
  \frac{\partial}{\partial t} q^{\natural}_t (s) \rb{1} \ve{1}_{s = t} +
  \dot{q}^{\natural}_t (t)$ since $q^{\natural}_t (t) = x$. Then $w (t) =
  \frac{d}{d t} \lp{1} q^{\natural}_t (t) \rp{1} = 0$. Moreover
  $q^{\natural}_t$, and therefore $w$, is solution of (\ref{lola12}) since the
  operator $\partial_t + \dot{q}^{\natural}_t (t) \cdot \partial_x$ does not
  depend on $s$. For small $t$, the null function is the unique solution of
  (\ref{lola12}) with vanishing boundary conditions. Then $w \equiv 0$ and
  (\ref{emma15}) is proven.
\end{proof}

\begin{remark}
  the identity (\ref{emma15}) is a generalization of [Ha4, (4.23)] and [Ha5,
  (3.18)].
\end{remark}

\begin{lemma}
  \label{emma15.2}Let $A$, B and $C$ as in Section \ref{lola} such that $A$,
  iB and $C$ satisfy (\ref{lola22}). Then there exists $T_e \in] 0, \bar{T} [$
  such that
  \begin{enumerate}
    \item \label{emma15.2.1}The matrix-valued functions $\tilde{q}^{\flat}_t$
    and $\tilde{q}^{\sharp}_t$ are real for $t \in i\mathbbm{R}, |t| < T_e$.
    
    \item \label{emma15.2.2}For $s, s' \in [0, 1]$, the coefficients of the
    matrix $\tilde{K}_t (s, s')$ are real for $t \in i\mathbbm{R}, |t| < T_e$.
    
    \item \label{emma15.2.3}There exist $\Phi_1$ a polynomial with respect to
    $x$ and $y$ whose coefficients are analytic near 0 and $k$ an analytical
    function near $0$ such that
    \[ p_t^0 (x, y) \assign \frac{k (t)}{(4 \pi \Delta t)^{\nu / 2}} e^{-
       \frac{1}{t} \Phi_1 (x, y, t)} \]
    and $\Phi_1 |_{x, y \in \mathbbm{R}^{\nu}, t \in i] - T_e, T_e [}$ is
    $\mathbbm{R}$-valued.
  \end{enumerate}
\end{lemma}

\begin{proof}
  We take the point of view of the Schr\"odinger equation. Since the functions
  $A, i B, C$ satisfy (\ref{lola22}), \ the Lagrangian $\tilde{L} |_{t = i
  \tilde{t}}$ is a polynomial with respect to $\dot{q}$ and $q$ with
  coefficients which are real functions with respect to $\tilde{t}$. Therefore
  the Euler equations associated to $\tilde{L} |_{t = i \tilde{t}}$ by
  differentiating with respect to $\tilde{t}$ have real coefficients and
  $\tilde{q}^{\flat}_{i \tilde{t}}$, $\tilde{q}^{\sharp}_{i \tilde{t}}$ and
  $\tilde{q}^{\natural}_{i \tilde{t}}$ (for $x, y \in \mathbbm{R}^{\nu}$) are
  real for $\tilde{t} \in \mathbbm{R}$, $| \tilde{t} |$ small enough. Then, by
  (\ref{lola18}), $\tilde{K}_t (s, s') \in \mathbbm{R}$ for $s, s' \in [0, 1]$
  and by (\ref{emma3}), $\Phi (x, y, t) \in \mathbbm{R}$ for $t \in
  i\mathbbm{R}$, $|t|$ small enough, and $x, y \in \mathbbm{R}^{\nu}$. Let us
  choose $\Phi_1 (t, x, y) \assign \Phi (t, x, y) - \Phi (t, y, y)$. Then by
  (\ref{emma8.5}), Assertion \ref{emma15.2.3} holds.
\end{proof}

\subsection{\label{emma15.5}The propagator equation}

In a heuristic way, the shape of the deformation formula can be explained by
the Wiener representation of the heat kernel and Wick's theorem (see [Ha4,
Appendix]). The matrix $K_t (s, s')$ appears therefore as a propagator. In
this section, we prove that $K_t (s, s')$ indeed satisfies the propagator
equation (cf. Proposition \ref{emma20}). First, we claim that $(q^{\flat}_t,
p^{\flat}_t)$ satisfies Hamiltonian equations associated to a Hamiltonian
\begin{equation}
  \label{emma16} \mathcal{H}= \tmop{Tr} ( \text{}^{\mathfrak{t}} p L p +
  \text{}^{\mathfrak{t}} p M q + \text{}^{\mathfrak{t}} q N q)
\end{equation}
where $L, M, N$ \ are matrix-valued analytic functions near $0$, $L$, $N$
being symmetric. Let us introduce some notation. Let $\mathcal{M}$ be a finite
dimensional square-matrix space. If $f$ is a regular $\mathbbm{R}$-valued
function on $\mathcal{M}^2$, we denote by $\frac{\partial f}{\partial X},
\text{$\frac{\partial f}{\partial Y}$}$ the matrices defined by
\[ d f (X, Y) \cdot (H, K) = \tmop{Tr} \lp{2} \text{}^{\mathcal{\mathfrak{t}}}
   H \frac{\partial f}{\partial X} \rp{2} + \tmop{Tr} \lp{2}
   \text{}^{\mathcal{\mathfrak{t}}} K \frac{\partial f}{\partial Y} \rp{2} .
\]
For instance $\frac{\partial f}{\partial X} = B Y$ and $\frac{\partial
f}{\partial Y} = \text{}^{\mathfrak{t}} B X$ if $f (X, Y) = \tmop{Tr} (
\text{}^{\mathfrak{t}} X B Y)$. This notation will allow us to take into
account the matrix structure of the trajectories of Lagrangian or Hamiltonian
systems. In particular, we shall use the matrix product (see Lemma
\ref{emma18}). The Euler-Lagrange equations associated to Lagrangian
\[ \mathcal{L} \assign \tmop{Tr} \lp{1} \frac{1}{4} \text{}^{\mathfrak{t}}
   \dot{q} A^{- 1} \dot{q} + \text{}^{\mathfrak{t}} \dot{q} B q +
   \text{}^{\mathfrak{t}} q C q \rp{1} \]
yields (\ref{lola14}), which proves that $q^{\flat}_t$ is a solution of these
equations. Therefore $p = \frac{\partial \mathcal{L}}{\partial \dot{q}} =
\frac{1}{2} A^{- 1} \dot{q} + B q$ and the Lagrangian $\mathcal{L}$ yields a
Hamiltonian
\[ \mathcal{H}= \tmop{Tr} \lp{2} \text{}^{\mathfrak{t}} \dot{q} \frac{\partial
   \mathcal{L}}{\partial \dot{q}} \rp{2} -\mathcal{L}= \tmop{Tr} \lp{1}
   \frac{1}{4} \text{}^{\mathfrak{t}} \dot{q} A^{- 1} \dot{q} -
   \text{}^{\mathfrak{t}} q C q \rp{1} = \tmop{Tr} \lp{1}
   \text{}^{\mathfrak{t}} (p - B q) A (p - B q) - \text{}^{\mathfrak{t}} q C q
   \rp{1} . \]
Then $(q^{\flat}_t, p^{\flat}_t)$ satisfies Hamiltonian equations with a
Hamiltonian as in (\ref{emma16}).

\begin{lemma}
  \label{emma18}Let $\mathcal{H}$ be a Hamiltonian as in (\ref{emma16}). Then
  the matrix $\text{}^{\mathfrak{t}} q p - \text{}^{\mathfrak{t}} p q$ is
  constant along Hamiltonian trajectories.
\end{lemma}

\begin{proof}
  It suffices to prove that $\text{}^{\mathfrak{t}} \dot{q} p +
  \text{}^{\mathfrak{t}} q \dot{p}$ is symmetric on Hamiltonian trajectories.
  Since $\dot{q} = \frac{\partial \mathcal{H}}{\partial p} = 2 L p + M q$ and
  $\dot{p} = - \frac{\partial \mathcal{H}}{\partial q} = - (2 N q +
  \text{}^{\mathfrak{t}} M p)$,
  \[ \text{$\text{}^{\mathfrak{t}} \dot{q} p + \text{}^{\mathfrak{t}} q
     \dot{p} = 2 \text{}^{\mathfrak{t}} p L p - 2 \text{}^{\mathfrak{t}} q N
     q$} . \]
  This matrix is symmetric which proves Lemma \ref{emma18}.
\end{proof}

Now, we can prove that $K_t (s, s')$ satisfies the propagator equation.

\begin{proposition}
  \label{emma20}For small positive number $t$ and every $(s, s') \in] 0, t
  [^2$, $K_t (s, s')$ satisfies
  \[ \left\{ \begin{array}{l}
       - \frac{d^2 K}{d s^2} + E (s) \frac{d K}{d s^{}} + F (s) K = A (s')
       \delta_{s = s'}\\
       \\
       K|_{s = 0} = K|_{s = t} = 0
     \end{array} . \right. \]
  
\end{proposition}

\begin{proof}
  By (\ref{lola16}) and since $q^{\flat}_s (s) = \mathbbm{1}$,
  \[ \frac{d K}{d s^{}} = \int_{s \vee s'}^t \dot{q}^{\flat}_{\tau} (s) A
     (\tau) \text{}^{\mathfrak{t}} q^{\flat}_{\tau} (s') d \tau - 1_{s' < s} A
     (s) \text{}^{\mathfrak{t}} q^{\flat}_s (s') \]
  and
  \[ \frac{d^2 K}{d s^2} = \int_{s \vee s'}^t \ddot{q}^{\flat}_{\tau} (s) A
     (\tau) \text{}^{\mathfrak{t}} q^{\flat}_{\tau} (s') d \tau \]
  \[ - 1_{s' < s} \lp{2} \dot{q}^{\flat}_s (s) A (s) \text{}^{\mathfrak{t}}
     q^{\flat}_s (s') + \frac{d}{d s} \lp{1} A (s) \text{}^{\mathfrak{t}}
     q^{\flat}_s (s') \rp{1} \rp{2} - A (s') \delta_{s = s'} . \]
  Then, since $q^{\flat}_{\tau}$ satisfies (\ref{lola14}),
  \[ - \frac{d^2 K}{d s^2} + E (s) \frac{d K}{d s^{}} + F (s) K = A (s')
     \delta_{s = s'} + 1_{s' < s} w (s') \]
  where
  \[ w (s') \assign \dot{q}^{\flat}_s (s) A (s) \text{}^{\mathfrak{t}}
     q^{\flat}_s (s') + \frac{d}{d s} \lp{1} A (s) \text{}^{\mathfrak{t}}
     q^{\flat}_s (s') \rp{1} - E (s) A (s) \text{}^{\mathfrak{t}} q^{\flat}_s
     (s') . \]
  We claim that $w \equiv 0$. Since $\text{}^{\mathfrak{t}} w$ satisfies
  (\ref{lola14}), it suffices to check that $w|_{s' = 0} = 0$ and $w|_{s' = s}
  = 0$. The first equality is obvious. Since $q^{\flat}_s (s) = \mathbbm{1}$,
  \[ \lb{2} \frac{\partial}{\partial s} q^{\flat}_s (s') \rb{2} \sve{2}{}{s' =
     s} = - \dot{q}^{\flat}_s (s) . \]
  Then
  \[ w|_{s' = s} = \dot{q}^{\flat}_s (s) A (s) + \dot{A} (s) - A (s)
     \text{}^{\mathfrak{t}} \dot{q}^{\flat}_s (s) - E (s) A (s) . \]
  By (\ref{lola13})
  \[ w|_{s' = s} = A (s) \lp{2} A^{- 1} (s) \dot{q}^{\flat}_s (s) -
     \text{}^{\mathfrak{t}} \dot{q}^{\flat}_s (s) A^{- 1} (s) + 2 \lp{1} B (s)
     - \text{}^{\mathfrak{t}} B (s) \rp{1} \rp{2} A (s) . \]
  But $p = \frac{\partial L}{\partial \dot{q}} = \frac{1}{2} A^{- 1} \dot{q} +
  B q$. Then
  \[ w|_{s' = s} = 2 A (s) \lp{1} p_s^{\flat} (s) - \text{}^{\mathfrak{t}}
     p_s^{\flat} (s) \rp{1} A (s) . \]
  By Lemma \ref{emma18}, the matrix $\text{}^{\mathfrak{t}} q^{\flat}_s
  p^{\flat}_s - \text{}^{\mathfrak{t}} p^{\flat}_s q^{\flat}_s$ is constant.
  It vanishes for $s' = 0$ and is equal to $p^{\flat}_s (s) -
  \text{}^{\mathfrak{t}} p^{\flat}_s (s)$ for $s' = s$. Hence $w|_{s' = s} =
  0$ and $w$ vanishes. This proves Proposition \ref{emma20}.
\end{proof}

\subsection{The deformation matrix}

For the proof of Theorem \ref{lola20}, we must establish some properties of
$\tmmathbf{\tilde{K}}_t (s)$ (here $s \in [0, 1]^n$, cf. (\ref{lola30})) which
we call the deformation matrix. We already studied it [Ha4] in a particular
case. The following lemma (see [Ha4]) will be useful.

\begin{lemma}
  \label{emma30}Let $\tilde{T} > 0$ and $M > 0$. There exists $T > 0$
  satisfying the following property. Let $f$ be an analytic function on
  $\disque{\tilde{T}}$ verifying $f (0) = 0$, $f' (0) = 1$, $\sup_{t \in
  \text{$\disque{\tilde{T}}$}} \text{$|f (t) | \leqslant M$}$ and, for every
  $t \in \disque{\tilde{T}}$,
  \begin{equation}
    \label{emma31} \mathcal{R} et = 0 \Rightarrow \mathcal{R} ef (t) = 0.
  \end{equation}
  Then, for $t \in \text{$\disque{T}$}$,
  \begin{equation}
    \label{emma32} \mathcal{R} et > 0 \Rightarrow \mathcal{R} ef (t) > 0.
  \end{equation}
\end{lemma}

\begin{proposition}
  \label{emma34}Let $\mathcal{E}$ be the space of measures $\mu = \sum_{j =
  1}^n \delta_{s_j} \xi_j$ such that $n \geqslant 1, \xi_j \in
  \mathbbm{R^{\nu}}, s_j \in] 0, 1 [$. For complex $t$ with small modulus, we
  denote by $(., .)_t$ \ the following bilinear form on $\mathcal{E}$
  \[ \text{$(\mu_1, \mu_2)_t$} \assign \int_0^1 \int_0^1 d \mu_1 (s) \cdot
     \tilde{K}_t (s, s') d \mu_2 (s') . \]
  Notice that
  \[ \text{$(\mu_1, \mu_2)_0$} = \int_0^1 \int_0^1 s \wedge s' (1 - s \vee s')
     A (0) .d \mu_1 (s) \otimes d \mu_2 (s') . \]
  Then for complex $t$ with small modulus
  \begin{equation}
    \label{emma36} \forall \mu \in \mathcal{E}, | (\mu, \mu)_t | \leqslant 2
    (\mu, \mu)_0 .
  \end{equation}
\end{proposition}

\begin{remark}
  \label{emma37}The bilinear form $(., .)_0$ is symmetric positive definite
  (see [Ha4, Rem. 4.4]).
\end{remark}

\begin{proof}
  By Proposition \ref{emma20} and analytic continuation, $\tilde{K}_t$
  satisfies for complex $t$ with small modulus
  \begin{equation}
    \label{emma38} \left\{ \begin{array}{l}
      A^{- 1} (t s) \lp{2} - \frac{d^2}{d s^2} + t E (t s) \frac{d}{d s^{}} +
      t^2 F (t s) \rp{2} \tilde{K}_t (s, s') = \delta_{s = s'}\\
      \\
      \tilde{K}_t |_{s = 0} = \tilde{K}_t |_{s = 1} = 0
    \end{array} . \right.
  \end{equation}
  Let $(\xi_1, \ldots, \xi_n) \in \mathbbm{R}^{\nu n}$ and $(s_1, \ldots, s_n)
  \in] 0, 1 [^n$. The function $u$ defined by
  \[ u (s) = \sum_{j = 1}^n  \tilde{K}_t (s, s_j) \xi_j  \]
  is continuous and piecewise differentiable on $[0, 1]$. Let $\mu \assign
  \sum_{j = 1}^n \delta_{s_j} \xi_j$. By (\ref{emma38})
  \begin{equation}
    \label{emma40} \left\{ \begin{array}{l}
      A^{- 1} (t s) \lp{2} - \frac{d^2}{d s^2} + t E (t s) \frac{d}{d s^{}} +
      t^2 F (t s) \rp{2} u = \mu\\
      \\
      u (0) = u (1) = 0
    \end{array} \right.
  \end{equation}
  Let $H^0 \assign \text{$L^2 ([0, 1], \mathbbm{C}^{\nu})$}$. Let
  $\varepsilon_{k, l}$ be the coordinates of an orthonormal basis of
  $\mathbbm{R}^{\nu}$ diagonalizing the real symmetric matrix $A^{- 1} (0)$.
  For $n \geqslant 1$, $k, l \in \{1, \ldots, \nu\}$ and $s \in [0, 1]$, set
  $e_{n, k, l} (s) = \sqrt{2} \sin (n \pi s) \varepsilon_{k, l}$. $(e_{n,
  k})_{n, k}$ is an orthonormal basis of $H^0$ which diagonalizes the
  unbounded self-adjoint operator $S \assign - A^{- 1} (0) \frac{d^2}{d s^2}$
  (Dirichlet boundary conditions). Let
  \begin{eqnarray*}
    H_0^1 \assign &  & \lbc{2} f \in H^0 \ve{2} \frac{d f}{d s} \in L^2, f (0)
    = f (1) = 0 \rbc{2}\\
    = &  & \text{$\lbc{2} \sum_{n, k} f_{n, k} e_{n, k}$} \ve{2} \sum_{n, k}
    |n f_{n, k} |^2 < \infty \rbc{2}
  \end{eqnarray*}
  and
  \[ H^{- 1} \assign \lbc{2} \sum_{n, k} f_{n, k} e_{n, k} \ve{2} \sum_{n, k}
     \ve{2} \frac{f_{n, k}}{n} \sve{2}{2}{} < \infty \rbc{2} . \]
  For $(f, g) \in H^{- 1} \times H_0^1$ or $(f, g) \in H^0 \times H^0$, let
  \[ \left\langle f, g \right\rangle \assign \int_0^1 \bar{f} (s) \cdot g (s)
     d s = \sum_{n, k} \bar{f}_{n, k} g_{n, k} . \]
  The operator $S^{1 / 2}$ can be viewed as an isomorphism from $H_0^1$ to
  $H^0$ and from $H^0$ to $H^{- 1}$. Natural Hilbertian norms induced by
  $S^{\pm 1 / 2}$ can be defined on $H_0^1$ and $H^{- 1}$. Then $\mu \in H^{-
  1}$ and
  \[ \mu = (S + T_t) u, \]
  where
  \[ T_t = - \lp{1} A^{- 1} (t s) - A^{- 1} (0) \rp{1}  \frac{d^2}{d s^2} +
     A^{- 1} (t s) \lp{1} t E (t s) \frac{d}{d s^{}} + t^2 F (t s) \rp{1} . \]
  Since
  \[ S + T_t = S^{1 / 2} (1 + S^{- 1 / 2} T_t S^{- 1 / 2}) S^{1 / 2} \]
  and $\|S^{- 1 / 2} T_t S^{- 1 / 2} ||_{L (H^0, H^0)}$ goes to $0$ when $t$
  goes to 0, one has, for small complex $t$, that $1 + S^{- 1 / 2} T_t S^{- 1
  / 2}$ is invertible and
  \[ \|(1 + S^{- 1 / 2} T_t S^{- 1 / 2})^{- 1} \|_{L (H^0, H^0)} \leqslant 2.
  \]
  Hence
  \begin{eqnarray*}
    (\mu, \mu)_t = &  & \left. \langle \mu, u \right\rangle\\
    = &  & \lba{1} \mu, S^{- 1 / 2} (1 + S^{- 1 / 2} T_t S^{- 1 / 2})^{- 1}
    S^{- 1 / 2} \mu \rba{1}\\
    = &  & \lba{1} S^{- 1 / 2} \mu, (1 + S^{- 1 / 2} T_t S^{- 1 / 2})^{- 1}
    S^{- 1 / 2} \mu \rba{1}
  \end{eqnarray*}
  and by Cauchy-Schwarz inequality
  \[ | (\mu, \mu)_t | \leqslant 2| S^{- 1 / 2} \mu |_{H^0}^2 . \]
  But $|S^{- 1 / 2} \mu |_{H^0}^2 = \lba{1} \mu, S^{- 1} \mu \rba{1} =
  \text{$(\mu, \mu)_0$}$. This proves (\ref{emma36}).
\end{proof}

\begin{proposition}
  \label{emma50}Let $A$, $B$ and $C$ be as in Theorem \ref{lola20}. There
  exists $T_d \in] 0, \bar{T} [$ such that for every $\text{$n \geqslant 1, s
  = (s_1, \ldots, s_n) \in [0, 1]^n, (\xi_1, \ldots, \xi_n) \in
  \mathbbm{R}^{\nu n}$}, t \in \mathbbm{C}$ with the condition $0 < s_1 <
  \cdots < s_n < 1, |t| < T_d$,
  \begin{equation}
    \label{emma52} \mathcal{R} et \geqslant 0 \Rightarrow \mathcal{R} e \lp{1}
    t \tmmathbf{\tilde{K}}_t (s) \cdot \xi \otimes_n \xi \rp{1} \geqslant 0,
  \end{equation}
  \begin{equation}
    \label{emma54} \ve{1} \tmmathbf{\tilde{K}}_t (s) \cdot \xi \otimes_n \xi
    \ve{1} \leqslant 2 n|A (0) |_{\infty} \sum_{j = 1}^n \xi_j^2 .
  \end{equation}
\end{proposition}

\begin{proof}
  Let $\mu \assign \sum_{j = 1}^n \delta_{s_j} \xi_j$. Then
  $\tmmathbf{\tilde{K}}_t (s) \cdot \xi \otimes_n \xi = (\mu, \mu)_t$ and
  \[ (\mu, \mu)_0 = \sum_{j, k = 1}^n s_{j \wedge k} (1 - s_{j \vee k}) \xi_j
     \cdot A (0) \xi_k . \]
  Then $(\mu, \mu)_0 \leqslant n|A (0) |_{\infty} \sum_{j = 1}^n \xi_j^2$.
  Hence Proposition \ref{emma34} implies (\ref{emma54}) if $T_d$ is small
  enough.
  
  Let us choose arbitrary vectors $\xi_1, \ldots, \xi_n$ such that $(\xi_1,
  \ldots, \xi_n)$ does not vanish. By Remark \ref{emma37}, $(\mu, \mu)_0 \neq
  0$. Let $f$ be the function defined by
  \[ f (t) = t \frac{(\mu, \mu)_t}{(\mu, \mu)_0} = t
     \frac{\tmmathbf{\tilde{K}}_t (s) \cdot \xi \otimes_n \xi}{(\mu, \mu)_0} .
  \]
  We claim that the function $f$ satisfies (\ref{emma31}). It suffices to
  check that $g (t) \assign \tmmathbf{\tilde{K}}_t (s) \cdot \xi \otimes_n
  \xi$ satisfies (\ref{lola22}). This holds by Lemma \ref{emma15.2}. By
  (\ref{emma36}), $f$ is bounded for complex $t$ with small modulus. Obviously
  $f (0) = 1$ and $f' (0) = 1$. By Lemma \ref{emma30}, there exists $T_d > 0$
  such that $\re f (t) > 0$ for $t \in \demidisque{T_d}$. Since $(\mu, \mu)_0
  \in] 0, + \infty [$, (\ref{emma52}) holds for $t \in \disque{T_d}$.
\end{proof}

\begin{remark}
  \label{emma56}The reality assumption (\ref{lola22}) is crucial for
  establishing (\ref{emma31}). What happens when (\ref{lola22}) does not hold?
  Then the statement of Lemma \ref{emma30} can be replaced by the following
  one. Let $\tilde{T} > 0$, $M > 0$ and $\varepsilon \in] 0, \pi / 2 [$. There
  exists $T_{\varepsilon} > 0$ such that every analytic function $f$ on
  $\disque{\tilde{T}}$, with $f (0) = 0$, $f' (0) = 1$, $\sup_{t \in
  \text{$\disque{\tilde{T}}$}} \text{$|f (t) | \leqslant M$}$, satisfies
  \[ \left\{ \begin{array}{l}
       t \in D_{T_{\varepsilon}}\\
       \arg t \in] - \pi / 2 + \varepsilon, \pi / 2 - \varepsilon [
     \end{array} \Rightarrow \re f (t) > 0. \right. \]
  Therefore (\ref{emma52}) can be replaced by
  \[ \left\{ \begin{array}{l}
       t \in D_{T_{\varepsilon}}\\
       \arg t \in] - \pi / 2 + \varepsilon, \pi / 2 - \varepsilon [
     \end{array} \Rightarrow \mathcal{R} e \lp{1} t \tilde{K}_t (s) \cdot \xi
     \otimes_n \xi \rp{1} \geqslant 0. \right. \]
  Then, even if Assumption (\ref{lola22}) is removed, the deformation formula
  will remain valid for $t \in D_{T_{\varepsilon}} -\{0\}, \arg t \in] - \pi /
  2 + \varepsilon, \pi / 2 - \varepsilon [$. One can expect to recover the
  Schr\"odinger kernel if the function $c$ is chosen as in [Ha7, Proposition
  4.5 (case $2$)].
\end{remark}

\subsection{Proof of Theorem \ref{lola20}}

The following lemma will be useful.

\begin{lemma}
  \label{emma68}Let $m \geqslant 0$ and $\Omega_1$, $\Omega_2$ be some open
  subsets of $\mathbbm{C}$ such that $\Omega_1 \subset \subset
  \Omega_2${\footnote{i.e. there exists $\rho > 0$ such that $\Omega_1 +
  \disque{\rho} \subset \Omega_2$.}}. There exists $C_{m, \Omega_1, \Omega_2}
  > 0$ satisfying the following property: for every analytic bounded
  matrix-valued function $\theta$ on $\Omega_2$ and every analytic bounded
  $\mathbbm{C}$-valued function $\varphi$ on $\Omega_2$ one has
  \begin{equation}
    \partial_t^m (\theta e^{\varphi}) = \alpha_m e^{\varphi}
  \end{equation}
  where $\alpha_m$ denotes an analytic matrix-valued function on $\Omega_2$
  such that
  \[ \sup_{\Omega_1} | \alpha_m | \leqslant C_{m, \Omega_1, \Omega_2} \times
     \sup_{\Omega_2} | \theta | \times \lp{1} 1 + (\sup_{\Omega_2} | \varphi
     |)^m \rp{1} . \]
\end{lemma}

\begin{proof}
  The lemma can be proved with the help of Cauchy's formula by induction on
  $m$.
\end{proof}

Let us prove Theorem \ref{lola20}. We choose $T_c = \frac{1}{2} \min (1, T_a,
T_b, T_d, T_e)$ (see Lemma \ref{emma9.1}, (\ref{lola21}), Proposition
\ref{emma50} and Lemma \ref{emma15.2}).

\bigskip

\tmtextbf{-1-} Let us check that $v_n$ given by (\ref{lola32}) and
$p^{\tmop{conj}} = \mathbbm{1} + \sum_{n \geqslant 1} v_n$ are well defined
for $t \in \demidisque{T_c}$. \ For $t \in \demidisque{2 T_c}$, let
\[ \varphi_n (t) \assign \tilde{q}^{\natural}_t (s_{}) \cdot \xi + i t
   \tmmathbf{\tilde{K}}_t (s) \cdot \xi \otimes_n \xi, \]
\[ F_n \assign t^n e^{i \varphi_n (t)} f (s_n t, \xi_n) \cdots f (s_1 t,
   \xi_1) . \]
Let $R > 0$ and let $(x, y) \in \mathbbm{C}^{2 \nu}$ such that $|x| + |y| <
R$. By Lemma \ref{emma9.1} and by (\ref{emma52}), $|i \tilde{q}^{\natural}_t
(s_{}) \cdot \xi | \leqslant 2 R (| \xi_1 | + \cdots + | \xi_n |)$ and $\re (t
\tmmathbf{\tilde{K}}_t (s) \cdot \xi \otimes_n \xi) \geqslant 0$. For
$\xi^{\ast} \in \mathbbm{R}^{\nu}$ and $\xi = (\xi_{_1}, \ldots, \xi_n) \in
\mathbbm{R}^{\nu n}$, let us denote
\[ \mathfrak{f}(\xi^{\ast}) \assign \sup_{|t| < T_b} |f (t, \xi^{\ast}) |,
   \mathfrak{f}^{\otimes} (\xi) \assign \mathfrak{f}(\xi_n) \cdots
   \mathfrak{f}(\xi_1), \]
\[ d^{\nu n} \mu_{\ast}^{\otimes} (\xi) \assign d \mu_{\ast} (\xi_n) \cdots d
   \mu_{\ast} (\xi_1), \]
\[ G_n \assign 2^n T_c^n \exp \lp{1} 2 R (| \xi_1 | + \cdots + | \xi_n |)
   \rp{1} \mathfrak{f}^{\otimes} (\xi) . \]
Then $|F_n | \leqslant G_n$. Let
\[ A \assign 2 \int_{\mathbbm{R}^{\nu}} \exp (2 R| \xi |)\mathfrak{f}(\xi) d
   \mu_{\ast} (\xi) . \]
Then
\[ \int_{\mathbbm{R}^{\nu n}} G_n d^{\nu n} \mu_{\ast}^{\otimes} (\xi)
   \leqslant (A T_c)^n \]
and
\[ \sum_{n \geqslant 1} \int_{0 < s_1 < \cdots < s_n < 1}
   \int_{\mathbbm{R}^{\nu n}} G_n d^{\nu n} \mu_{\ast}^{\otimes} (\xi) d^n s <
   \infty . \]
Hence $v_n$ and $p^{\tmop{conj}}$ are well defined on $\demidisque{2 T_c}
\times \mathbbm{C}^{2 \nu}$ since $R$ is arbitrary (let us remark that the
expressions (\ref{lola24}) and (\ref{lola32}) of $v_n$ are clearly
equivalent). By the dominated convergence theorem, $p^{\tmop{conj}} \in
\mathcal{A} ( \demidisque{T_c} \times \mathbbm{C}^{2 \nu})$.

\bigskip

\tmtextbf{-2-} Let us check that $p^{\tmop{conj}}$ is well defined for $t \in
\bar{D}^+_{T_c}$ and that $p^{\tmop{conj}} \in \mathcal{C}^{\infty} \lp{1}
\bar{D}^+_{T_c}, \mathcal{A}(\mathbbm{C}^{2 \nu}) \rp{1}$. \ Let $R > 0$. By
(\ref{emma10}) and (\ref{emma54}), there exists $M > 0$ such that for $n
\geqslant 1$, $0 < s_1 < \cdots < s_n < 1$, $(\xi_1, \ldots, \xi_n) \in
\mathbbm{R}^{\nu n}$, $t \in \disque{2 T_c}$, $(x, y) \in \mathbbm{C}^{2 \nu}$
and $|x| + |y| < R$,
\begin{equation}
  \label{isabella10} | \varphi_n (t) | \leqslant M \lp{1} 1 + n (| \xi_1 | +
  \cdots | \xi_n |)^2 \rp{1} .
\end{equation}
We want to use the dominated convergence theorem. For $m \geqslant 0$, let
\[ F_{n, m} \assign \partial_t^m \lp{1} t^n f (s_n t, \xi_n) \cdots f (s_1 t,
   \xi_1) e^{i \varphi_n (t)} \rp{1} . \]
By Lemma \ref{emma68} and (\ref{isabella10}), there exists $K_m > 0$ such that
\[ | F_{n, m} | \leqslant K_m^{n + 1} n^m \lp{1} 1 + (| \xi_1 | + \cdots + |
   \xi_n |)^{2 m} \rp{1} \exp \lp{1} 2 R (| \xi_1 | + \cdots + | \xi_n |)
   \rp{1} \mathfrak{f}^{\otimes} (\xi), \]
for $t \in \disque{T_c}$, $\re t \geqslant 0$. Then the inequality
\[ \frac{1}{(2 m) !} \lp{1} 1 + (| \xi_1 | + \cdots + | \xi_n |)^{2 m} \rp{1}
   \leqslant \exp (| \xi_1 | + \cdots + | \xi_n |) \]
yields $|F_{n, m} | \leqslant G_{n, m}$ where
\[ G_{n, m} \assign (2 m) !K_m^{n + 1} n^m \exp \lp{1} (1 + 2 R) (| \xi_1 | +
   \cdots + | \xi_n |) \rp{1} \mathfrak{f}^{\otimes} (\xi) . \]
Let
\[ A = \int_{\mathbbm{R}^{\nu}} \exp \lp{1} (1 + 2 R) | \xi | \rp{1}
   \mathfrak{f}(\xi) d \mu_{\ast} (\xi) . \]
Then
\[ \int_{\mathbbm{R}^{\nu n}} G_{n, m} d^{\nu n} \mu_{\ast}^{\otimes} (\xi)
   \leqslant (2 m) !n^m K_m^{n + 1} A^n \]
and
\[ \sum_{n \geqslant 1} \int_{0 < s_1 < \cdots < s_n < 1}
   \int_{\mathbbm{R}^{\nu n}} G_{n, m} d^{\nu n} \mu_{\ast}^{\otimes} (\xi)
   d^n s < \infty . \]
Since $R$ and $m$ are arbitrary, the dominated convergence theorem proves that
$p^{\tmop{conj}} \in \text{$\mathcal{C}^{\infty} \lp{1} \bar{D}^+_{T_c},
\mathcal{A}(\mathbbm{C}^{2 \nu}) \rp{1}$}$.

\bigskip

\tmtextbf{-3-} Let us check that $p^{\tmop{conj}} \in \mathcal{C}^{\infty}_{b,
1} \lp{1} i] - T_c, T_c [\times \mathbbm{R}^{2 \nu} \rp{1}$. Let $\alpha \in
\mathbbm{N}, \beta, \gamma \in \mathbbm{N}^{\nu}, x, y \in \mathbbm{R}^{\nu}$
and $\tilde{t} \in] - 2 T_c, 2 T_c [$. Let
\[ F_n \assign \partial^{\alpha}_{\tilde{t}} \partial^{\beta}_x
   \partial^{\gamma}_y \lp{1} (i \tilde{t} \srp{0}{n}{} e^{i \varphi_n (i
   \text{$\tilde{t}$})} f (i \tilde{t} s_n, \xi_n) \cdots f (i \tilde{t} s_1,
   \xi_1) \rp{1} . \]
Let $(e_1, \ldots, e_{\nu})$ be the standard basis of $\mathbbm{R}^{\nu}$. For
$\delta = 1, \ldots, \nu$, let us denote
\[ \varpi_{\delta, \flat} (t) \assign ( \tilde{q}^{\flat}_t (s_1) e_{\delta})
   \cdot \xi_1 + \cdots + ( \tilde{q}^{\flat}_t (s_n) e_{\delta}) \cdot \xi_n,
\]
\[ \varpi_{\delta, \sharp} (t) \assign ( \tilde{q}^{\sharp}_t (s_1)
   e_{\delta}) \cdot \xi_1 + \cdots + ( \tilde{q}^{\sharp}_t (s_n) e_{\delta})
   \cdot \xi_n . \]
Then
\[ F_n = i^{\alpha} \partial^{\alpha}_t \lp{1} \theta_n (t) e^{i \varphi_n
   (t)} \rp{1} \sve{1}{}{t = i \text{$\tilde{t}$}} \]
where
\[ \theta_n (t) \assign t^n \varpi^{\beta_1}_{1, \flat} (t) \cdots
   \varpi^{\beta_{\nu}}_{\nu, \flat} (t) \varpi^{\gamma_1}_{1, \sharp} (t)
   \cdots \varpi^{\gamma_{\nu}}_{\nu, \sharp} (t) f (t s_1, \xi_1) \cdots f (t
   s_n, \xi_n) . \]
Then there exists $M_1 > 0$ such that, for $t \in D_{2 T_c}$,
\[ | \theta_n (t) | \leqslant M_1 | \xi |_1^{| \beta | + | \gamma |}
   \mathfrak{f}^{\otimes} (\xi) . \]
By (\ref{emma54}) there exists $M_2 > 0$, such that, for $t \in D_{2 T_c}$,
\[ | \varphi_n (t) | \leqslant M_2 \lp{1} (|x| + |y|) | \xi |_1 + n| \xi |_1^2
   \rp{1} . \]
Then, by Lemma \ref{emma68}, there exists $C > 0$ such that for $\tilde{t}
\in] - T_c, T_c [$, $n \geqslant 1$, $\xi \in \mathbbm{R}^{\nu n}$ and $x, y
\in \mathbbm{R}^{\nu}$
\[ | F_n | \leqslant C \lp{2} 1 + \lp{1} (|x| + |y|) | \xi |_1 + n| \xi |_1^2
   \srp{1}{\alpha}{} \rp{2} | \xi |_1^{| \beta | + | \gamma |}
   \mathfrak{f}^{\otimes} (\xi) . \]
Here we also use that, by assertions \ref{emma15.2.1} and \ref{emma15.2.2} of
Lemma \ref{emma15.2}, $\varphi_n (i \text{$\tilde{t}$}) \in \mathbbm{R}$.
Then, by binomial formula, there exits $C' > 0$ such that
\begin{eqnarray*}
  | F_n | \leqslant &  & C' \times \lp{2} 1 + \sum_{\alpha_1 + \alpha_2 =
  \alpha} (|x| + |y|)^{\alpha_1} | \xi |_1^{\alpha_1 + 2 \alpha_2}
  n^{\alpha_2} \rp{2} \times | \xi |_1^{| \beta | + | \gamma |}
  \mathfrak{f}^{\otimes} (\xi)\\
  \leqslant &  & Q (x, y) n^{\alpha} e^{| \xi |_1} \mathfrak{f}^{\otimes}
  (\xi)
\end{eqnarray*}
where
\[ Q (x, y) \assign C' \lp{2} 1 + \sum_{\alpha_1 + \alpha_2 = \alpha} \lp{1}
   \alpha_1 + 2 \alpha_2 + | \beta | + | \gamma | \rp{1} ! (|x| +
   |y|)^{\alpha_1} \rp{2} . \]
Let
\[ A = \int_{\mathbbm{R}^{\nu}} e^{| \xi |} \mathfrak{f}(\xi) d \mu_{\ast}
   (\xi) . \]
Then
\[ \int_{\mathbbm{R}^{\nu n}} |F_n |d^{\nu n} \mu_{\ast}^{\otimes} (\xi) d^n s
   \leqslant Q (x, y) n^{\alpha} A^n . \]
Therefore, for $t \in i] - T_c, T_c [$ and $x, y \in \mathbbm{R}^{\nu}$,
\[ | \partial^{\alpha}_t \partial^{\beta}_x \partial^{\gamma}_y u| \leqslant 1
   + Q (x, y) \sum_{n \geqslant 1} n^{\alpha} \frac{A^n}{n!} . \]
This proves that $p^{\tmop{conj}} \in \mathcal{C}^{\infty}_{b, 1} (i] - T_c,
T_c [\times \mathbbm{R}^{2 \nu})$.

\bigskip

\tmtextbf{-4-} Let us verify that the function $p^0 p^{\tmop{conj}}$, with
$p^{\tmop{conj}}$ given by (\ref{lola23}), is a solution of (\ref{ella6}). By
continuity and analyticity arguments, it suffices to check (\ref{ella6}) for
small positive number $t$. Let
\[ D_x \assign \partial_x + B (t) x. \]
Then, if $v = v (t, x)$ is a regular function with respect to its arguments,
\[ A (t) \cdot D_x^2 (p^0 v) = \lp{1} A (t) \cdot D_x^2 p^0 \rp{1} v + 2 A (t)
   \cdot D_x p^0 \otimes \partial_x v + p^0 A (t) . \partial_x^2 v. \]
A solution $u$ of (\ref{ella6}) is then given, if we use the relation $u = p^0
v$, by a solution $v$ of the conjugate equation
\[ \left\{ \begin{array}{l}
     \lp{1} \text{$\partial_t - \frac{2}{p^0} A (t) \cdot D_x p^0 \otimes
     \partial_x \rp{1} v = A (t) \cdot \partial_x^2 v$} + c (t, x) v \text{ \
     \ \ \ (} t \neq 0 \text{)}\\
     \\
     v_{} |_{t = 0^+} =_{} \mathbbm{1}
   \end{array} . \right. \]
Let $v_0 = \mathbbm{1} .$ By (\ref{emma14}), it suffices to verify that, for
$n \geqslant 1$, $v_n$ given by (\ref{lola24}) satisfies
\begin{equation}
  \label{emma70} \left\{ \begin{array}{l}
    \lp{1} \partial_t + \dot{q}^{\natural}_t (t) \cdot \partial_x \rp{1} v_n
    \text{$= A (t) \cdot \partial^2_x v_n$} + c (t, x) v_{n - 1}\\
    \\
    v_n |_{t = 0^+} =_{} 0
  \end{array} \right.,
\end{equation}
for small positive number $t$ and $n \geqslant 1$. By (\ref{lola24})
\[ v_n = \int_{0 < s_1 < \cdots < s_n < t} F_n d^n s, \]
where
\[ F_n = \lb{2} \exp \lp{1} \tmmathbf{K}_t (s) \cdot \partial_z \otimes_n
   \partial_z \rp{1} c (s_n, z_n) \cdots c (s_1, z_1) \rb{2}
   \ve{2}_{\tmscript{\begin{array}{l}
     z_1 = q^{\natural}_t (s_1)\\
     \ldots\\
     z_n = q^{\natural}_t (s_n)
   \end{array}}} . \]
Here
\[ \tmmathbf{K}_t (s) \cdot \partial_z \otimes_n \partial_z \assign \sum_{j,
   k = 1}^n \partial_{z_j} \cdot K_t (s_j, s_k) \partial_{z_k} \]
where $K_t$ is defined by (\ref{lola16}). Then
\[ \lp{1} \partial_t + \dot{q}^{\natural}_t (t) \cdot \partial_x \rp{1} v_n =
   (\tmop{boundary}) + (\tmop{interior}) \]
where
\[ (\tmop{boundary}) = \int_{0 < s_1 < \cdots < s_{n - 1} < t} F_n |_{s_n = t}
   d^{n - 1} s, \]
\[ (\tmop{interior}) = \int_{0 < s_1 < \cdots < s_n < t} \lp{1}
   \text{$\partial_t + \dot{q}^{\natural}_t (t) \cdot \partial_x \rp{1} F_n$}
   d^n s. \]
Since
\[ \tmmathbf{K}_t (s_1, \ldots, s_{n - 1}, t) \cdot \partial_z \otimes_n
   \partial_z = \tmmathbf{K}_t (s_1, \ldots, s_{n - 1}) \cdot \partial_z
   \otimes_{n - 1} \partial_z, \]
one gets
\[ \text{$(\tmop{boundary}) = c (t, x) v_{n - 1}$.} \]
Now we claim that $(\tmop{interior}) = A (t) \cdot \partial^2_x v_n$. By
(\ref{emma15}), if $\varphi (z_1, \ldots, z_n$) is an arbitrary differentiable
function of $(z_1, \ldots, z_n) \in \mathbbm{C}^{\nu n}$,
\[ \lp{1} \partial_t \text{$+ \dot{q}^{\natural}_t (t) \cdot \partial_x
   \rp{1}$} \text{$[\varphi (z_1, \ldots, z_n$)]$|_{\tmscript{\begin{array}{l}
     z_1 = q^{\natural}_t (s_1)\\
     \ldots\\
     z_n = q^{\natural}_t (s_n)
   \end{array}}}$} = 0. \]
Then
\[ (\tmop{interior}) = \int_{0 < s_1 < \cdots < s_n < t} G_n d^n s, \]
where
\[ G_n = \lb{2} \partial_t \lp{1} \tmmathbf{K}_t (s) \cdot \partial_z
   \otimes_n \partial_z \rp{1} \exp \lp{1} \tmmathbf{K}_t (s) \cdot \partial_z
   \otimes_n \partial_z \rp{1} c (s_n, z_n) \cdots c (s_1, z_1) \rb{2}
   \ve{2}_{\tmscript{\begin{array}{l}
     z_1 = q^{\natural}_t (s_1)\\
     \ldots\\
     z_n = q^{\natural}_t (s_n)
   \end{array}}} . \]
On the other hand,
\[ A (t) \cdot \partial^2_x v_n = \int_{0 < s_1 < \cdots < s_n < t} H_n d^n s
\]
where
\[ H_n = A (t) \cdot \partial^2_x \lb{2} \exp \lp{1} \tmmathbf{K}_t (s) \cdot
   \partial_z \otimes_n \partial_z \rp{1} c (s_n, z_n) \cdots c (s_1, z_1)
   \rb{2} \ve{2}_{\tmscript{\begin{array}{l}
     z_1 = q^{\natural}_t (s_1)\\
     \ldots\\
     z_n = q^{\natural}_t (s_n)
   \end{array}}} . \]
For $\bar{s} \in [0, t]$,  $q^{\natural}_t ( \bar{s}) = q^{\flat}_t ( \bar{s})
x + q^{\sharp}_t ( \bar{s}) y$. Then
\[ H_n = \lb{3} \sum_{j, k = 1}^n A (t) \cdot \lp{1} \text{}^{\mathfrak{t}}
   q^{\flat}_t (s_j) \partial_{z_j} \otimes \text{}^{\mathfrak{t}} q^{\flat}_t
   (s_k) \partial_{z_k} \rp{1} \exp \lp{1} \tmmathbf{K}_t (s) \cdot \partial_z
   \otimes_n \partial_z \rp{1} \times \]
\[ c (s_n, z_n) \cdots c (s_1, z_1) \rb{3} \ve{3}_{\tmscript{\begin{array}{l}
     z_1 = q^{\natural}_t (s_1)\\
     \ldots\\
     z_n = q^{\natural}_t (s_n)
   \end{array}}} . \]

By (\ref{lola16}), $\partial_t K_t (s, s') = q^{\flat}_t (s) A (t)
\text{}^{\mathfrak{t}} q^{\flat}_t (s')$ and then
\[ \partial_t \lp{1} \tmmathbf{K}_t (s) \cdot \partial_z \otimes_n \partial_z
   \rp{1} = \sum_{j, k = 1}^n \partial_{z_j} \cdot q^{\flat}_t (s_j) A (t)
   \text{}^{\mathfrak{t}} q^{\flat}_t (s_k) \partial_{z_k} . \]
Therefore $G_n = H_n$ and
\[ (\tmop{interior}) = A (t) \cdot \partial^2_x v_n . \]
Then (\ref{emma70}) holds and $p^0 p^{\tmop{conj}}$ satisfies (\ref{ella6}).

\bigskip

\ \ \ \ \ \ \ \ \ \ \ \ \ \ \ \ \ \ \ \ \ \ \ \ \ \ \ \ \ \ \ \ \ \ \ \ \ \ \
\ REFERENCES

\medskip

[A-H] S. A. Albeverio, R. J.Hoegh-Krohn, \tmtextit{\tmtextup{Mathematical
Theory of Feynman path integrals}}, Lecture Notes in Mathematics
\tmtextbf{523} (1976).

[B-B] R. Balian and C. Bloch, \tmtextup{Solutions of the Schr\"odinger
equation in terms of classical paths}, Ann. of Phys. \tmtextbf{85} (1974),
514-545.

[Ge-Ya] I.M. Gel'fand, A.M. Yaglom, Integration in functional spaces and its
applications in quantum physics, Journal of Mathematical Physics, 1-1 (1960),
48-69.

[Ha4] T. Harg\'e, Borel summation of the small time expansion of the heat
kernel. The scalar potential case (2013).

[Ha7] Some remarks on the complex heat kernel on $\mathbbm{C}^{\nu}$ in the
scalar potential case (2013).

[Ha5] T. Harg\'e, Borel summation of the heat kernel with a vector potential
case (2013).

[It] K. Ito, \tmtextup{Generalized uniform complex measures in the Hilbertian
metric space with their applications to the Feynman integral}, Fifth berkeley
Symp. on Math. Statist. and Prob. \tmtextbf{2} (1967), 145--161.

[On] E. Onofri, On the high-temperature expansion of the density matrix,
American Journal Physics, 46-4 (1978), 379-382.

[Ru] W. Rudin, Real and complex analysis, section 6.

\bigskip

D\'epartement de Math\'ematiques, Laboratoire AGM (CNRS), Universit\'e de
Cergy-Pontoise, 95000 Cergy-Pontoise, France.

\end{document}